\documentclass[letterpaper, 10 pt, conference]{ieeeconf} 
\usepackage{amsthm}

\usepackage{cite}
\usepackage[cmex10]{amsmath}
\usepackage{algorithm}
\usepackage{algorithmic}
\usepackage{breqn}

\usepackage[inline]{trackchanges}

\usepackage{amsmath}
\usepackage{amsfonts}
\usepackage{amssymb}
\usepackage{mathtools}

\addeditor{swanand}
\addeditor{alex}
\addeditor{brenden}
\addeditor{salim}

\usepackage{amsfonts, amssymb, amscd, xspace}

\usepackage{algorithm}
\usepackage{algorithmic}

\usepackage{color}
\usepackage{graphicx}
\usepackage{epstopdf}

\usepackage{xfrac}

\newtheorem{theorem}{Theorem}
\newtheorem{lemma}{Lemma}

\newtheorem{corollary}{Corollary}

\newtheorem{example}{Example}
\newtheorem{definition}{Definition}
\newtheorem{proposition}{Proposition}

\newcommand{\ie}{{\it i.e.}}

\newcommand{\Xj}[1]{X_{#1}} 

\newcommand{\Q}[2]{Q^{[#1, #2]}} 
\newcommand{\A}[2]{A^{[#1,#2]}} 

\newcommand{\Hp}[1]{H\left(#1\right)} 
\newcommand{\Hc}[2]{H\left(#1 \mid #2\right)} 
\newcommand{\I}[2]{I\left(#1\: ; \: #2\right)} 

\newcommand{\GF}[1]{\mathbb{F}_{#1}} 

\newcommand\blfootnote[1]{%
  \begingroup
  \renewcommand\thefootnote{}\footnote{#1}%
  \addtocounter{footnote}{-1}%
  \endgroup
}

\renewcommand{\baselinestretch}{1.015}

\title{\LARGE \bf
Private Information Retrieval with Side Information}

\IEEEoverridecommandlockouts

\begin{document}



%
\author{Swanand Kadhe, Brenden Garcia, Anoosheh Heidarzadeh, Salim El Rouayheb,  and
 Alex Sprintson}



\maketitle

\begin{abstract}
We study  the problem of Private Information Retrieval (PIR) in the presence of prior side information. The problem setup includes a database of  $K$ independent messages possibly replicated on several servers, and a user that needs to retrieve one of these messages. In addition, the user has some prior side information in the form of a subset of  $M$ messages, not containing the desired message and unknown to the servers. This problem is motivated by practical settings in which the user can obtain side information  opportunistically from other users or has   previously downloaded some messages using classical PIR schemes. The objective of the user is to retrieve the required message without revealing its identity while minimizing the amount of data downloaded from the servers.  

We focus on achieving information-theoretic privacy in  two  scenarios: (i) the user wants to protect jointly its demand and side information; (ii) the user wants to protect only the information  about its demand, but not the side information. To highlight the role of side information, we focus first on the case of a single server (single database). In the first scenario, we prove that the minimum  download cost is $K-M$ messages, and in the second scenario it is  $\lceil \frac{K}{M+1}\rceil$ messages, which should be compared to $K$ messages, the minimum download cost in the case of no side information. Then, we extend  some of our  results to the case of the database replicated on  multiple servers.  
Our proof techniques relate  PIR  with side information to the    index coding problem. We leverage this connection  to prove converse results, as well as to design  achievability schemes. 
\end{abstract}

\blfootnote{Swanand Kadhe, Brenden Garcia, Anoosheh Heidarzadeh, and Alex Sprintson are with the Department of Electrical and Computer Engineering at Texas A\&M University, USA; emails:\{swanand.kadhe,brendengarcia,anoosheh,spalex\}@tamu.edu. 

Salim El Rouayheb is with ECE Department  at Rutgers University, email: sye8@soe.rutgers.edu. Part of this work was done while he was with the ECE department at the Illinois Institute of Technology.

 The work of S. El Rouayheb  was supported in part by NSF Grant CCF 1652867 and ARL Grant W911NF-17-1-0032.}

\section{Introduction}
\label{sec:intro}

Consider the following Private Information Retrieval (PIR) setting first studied in \cite{Chor:PIR1995,chor1998private}:  a user  wishes to privately download  a message belonging to a database with copies stored on a single or   multiple remote servers, without revealing which message it is requesting. 
In a straightforward PIR scheme, the user would  download all the messages in  the database. This scheme may not be feasible due to the its high communication cost. In the case of a single server (i.e., there is only one copy of the database),  it can be shown that downloading  the whole database is necessary to achieve perfect privacy in an information-theoretic sense. If computational (cryptographic) privacy is desired, then PIR schemes with lower communication overhead do exist \cite{kushilevitz1997replication, cPIRPoly}, but they do not offer information-theoretic privacy guarantees and usually have high computational complexity. In contrast, in this paper, we design and analyze schemes that achieve information-theoretic privacy.

Interestingly, more efficient PIR schemes, achieving perfect  information-theoretic privacy,  can be constructed when the database is replicated on multiple servers with restriction on the servers' collusion.
 This replication-based model  has been the one that is predominantly studied in the PIR literature, with  breakthrough results in the past few years 
 (e.g., \cite{sun2016capacitynoncol, sun2016capacity, yekhanin2010private, beimel2001information, beimel2002breaking,gasarch2004survey}). 
Recently, there has been a renewed  interest in  PIR for the case in which the data is  stored on the servers using erasure codes, which result in  better storage overhead compared to the traditional replication techniques \cite{shah2014one, chan2014private, tajeddine2016private, extended, banawan2016capacity, fazeli2015pir, blackburn2016pir, freij2016private}.


%

%

In this paper, we study the  PIR problem when  the user has  prior side information about the database. In particular, we assume that the user already has a random subset of the database messages that is unknown to the server(s)\footnote{We assume that this side information subset does not contain the desired message. Otherwise, the problem is degenerate.}. This  side information could have been obtained in several ways. For example, the user could have obtained  these messages opportunistically from other users in its network,  overheard them from a wireless broadcast channel, or downloaded them previously through classical PIR schemes. The next example illustrates how this side information could be leveraged to devise efficient  PIR schemes. In particular, 
the following example shows that perfect information-theoretic privacy can be achieved with  a single server case without having to download the entire database.

\begin{example}[single-server  PIR with side information]\label{ex:intro}
Consider a remote server that has a database formed of  an even number of binary messages  denoted by $X_1,\dots,X_K$ of equal length.  A user wants to download one of these messages from the server without revealing to the server which one. Moreover, the user has one  message as side information chosen uniformly at random among all the other messages and unknown to the server. We propose two PIR schemes that leverage the side information  and compare them to the straightforward scheme that downloads all the $K$ messages.
\begin{enumerate}

\item{\em Maximum Distance Separable (MDS) PIR scheme.} This scheme downloads $K-1$ messages. The user sends to the server the number of messages in its side information  (one in this example). The server responds by coding all the messages using a $(2K-1,K)$ systematic MDS code and 
sending  the $K-1$ parity symbols of the code. Therefore, the user can always decode all the messages using its side information and the coded messages received from the server.

\item{\em Partition and Code PIR scheme.} This scheme downloads $K/2$ messages. Suppose the message the user wants is $X_W$ and the one in its side information is $X_S$ for some $W,S \in\{1,\dots,K\}$, $W\neq S$. The user chooses  a random partition of $\{1,\dots,K\}$ formed only of sets of size $2$ and containing $\{W,S\}$, and sends indices of all pairs in the partition to the server. The server sends back the XOR of the messages indexed by each subset. For example, suppose  $W=1$ and  $S=2$, i.e, the user wants $X_1$ and has $X_2$ as side information.  The user chooses a random partition $\{\{i_1,i_2\},\{i_3,i_4\},\dots,\{i_{K-1},i_K\}\}$ and sends it to the server. The partition is chosen such that $\{1,2\}$ is a part of the partition (i.e., $i_j=1$ and $i_{j+1}=2$ for some $j\in\{1,3,\dots,K-1\}$. 
The server responds with $X_{i_1}+X_{i_2},\dots, X_{i_{K-1}}+X_{i_K}$. The user can always decode because it always receives $X_W+X_S$. Intuitively, perfect privacy is achieved here because the index of the desired message can be in any subset of the partition, and in each subset it could be either one of messages in the subset, since the server does not know the index of the side information. \hfill\rule{1.3ex}{1.3ex}
 \end{enumerate}
 \end{example}
We will show later that the two schemes above are optimal but achieve different privacy constraints. The MDS PIR scheme protects both the indices of the desired message and that of the side information, whereas the Partition and Code scheme protects only the former.

\subsection{Our Contributions}
\label{sec:contributions}
We consider the PIR with side information problem as illustrated in  Example~\ref{ex:intro}. A user wishes to download a message from a set of $K$ messages that belong to a database stored on a single remote server or replicated on several {\em non-colluding} servers. 
The user has  a random subset of $M$ messages as side information. The identity of the messages in this subset is unknown to the server. We focus on PIR schemes that achieve information-theoretic privacy. The figure of merit that we consider for  the PIR schemes is the download rate, which dominates the total communication rate (download plus upload) for large message sizes. Under this setting, we distinguish between two types of privacy constraints: (i) hiding both the  identity of the requested message and that of  the side information from the server; and (ii) hiding only the identity of the desired message. The latter, and  less stringent, privacy constraint is justified when the side information is obtained opportunistically given that  it is random and assumed to be independent of the user's request. In the case in which  the side information messages were  obtained previously through PIR, this constraint implies that the identity of these messages may be leaked to the server(s). However, this type of privacy can still be relevant when privacy is only desired for a certain duration of time, i.e., when the user is  ambivalent about protecting the identity of messages downloaded as long as it has happened far enough in the past.  

First, we focus on the  single server scenario as the canonical case to understand the role of side information in PIR. We  characterize the   capacity of PIR with side information in the case of a single server for the two privacy constraints mentioned above. We show that when protecting both the side information and the request, the minimum download rate\footnote{The download rate is defined as the inverse of the normalized download cost.} for PIR is $(K-M)^{-1}$, and this can be achieved by a generalization of the MDS PIR scheme in Example~\ref{ex:intro}. Moreover, we show that when only protecting the request, the minimum download rate is $\lceil \frac{K}{M+1}\rceil^{-1}$, and this can be achieved by a generalization of the Partition and Code PIR scheme in Example~\ref{ex:intro}. We present achievability and converse proofs that use among others connections to index coding. Second, we tackle the case of $N>1$  servers storing replicas of the database. In this case, when $(M+1)\mid K$, we devise a PIR scheme with side information that achieves a download rate equal to  $$\left(1 + \frac{1}{N} + \cdots + \frac{1}{N^{\frac{K}{M+1}-1}}\right)^{-1}.$$ Our scheme for the multiple servers uses ideas from the single server scheme in conjunction with the scheme due to Sun and Jafar \cite{sun2016capacitynoncol} for settings with no side information.



\subsection{Related Work} \label{sec:related-work}
The initial work on PIR in \cite{Chor:PIR1995,chor1998private} and in the literature that followed  focused on  designing PIR schemes for replicated data that have   efficient  communication cost accounting  for  both the size of the user queries and the servers' responses. PIR schemes with communication cost that is subpolynomial in the number of messages were devised in \cite{beimel2002breaking} and \cite{dvir20162}. Information-theoretic bounds on the download rate (servers' responses) and achievable schemes were  devised in \cite{sun2016capacitynoncol} and \cite{sun2016capacity}. 
Recently, there has been   a growing body of work studying   PIR for coded data motivated by lower overhead of codes
 \cite{shah2014one, chan2014private, tajeddine2016private, extended, banawan2016capacity, fazeli2015pir, blackburn2016pir, freij2016private,tajeddine2017private1,tajeddine2017private2}.

The role of side information in improving PIR schemes has so far received little attention in the literature. The closest work to ours is the concurrent work  of Tandon \cite{Tandon2017} in which the capacity of  {cache-aided PIR} is characterized. The main difference with the model in \cite{Tandon2017} is our assumption that the indices of the side information messages are unknown to  the servers, as is the case in the scenarios mentioned above. This lack of knowledge at the servers can be leveraged to reduce the communication cost of PIR even in the case of a single server. We also restrict our study to side information that is subset of the data, whereas the cache model in \cite{Tandon2017} allows any function of the data.  Another related line of work is that of {private broadcasting} by Karmoose et al.\cite{Karmoose2017}, which considers the index coding setting with multiple users with side information and a single server. Here too, the server does know the content of the side information at the users. Moreover, the privacy constraint  is to protect the request and side information of a user from the other users through a carefully designed encoding matrix. In contrast, the goal of our scheme is to protect the identity of the requested data from the server. We also note that the case in which the side information is unknown at the server was also considered in the index coding literature under the name of  {blind index coding} \cite{kao2017blind}. However, the goal there was to minimize the broadcast rate without privacy constraints. 

\section{Problem Formulation and Main Results}
\label{sec:basics}

For a positive integer $K$, denote $\{1,\dots,K\}$ by $[K]$. 
For a set $\{X_1,\dots,X_K\}$ and a subset $S\subset {[K]}$, let \mbox{$X_S = \{X_j : j\in S\}$}. For a subset $S \subset [K]$, let $\mathbf{1}_S$ denote the characteristic vector of the set $S$, which is a binary vector of length $K$ such that, for all $j\in[K]$, its $j$-th entry is $1$ if $j\in S$, otherwise it is $0$. Let $\GF{q}$ denote the finite field of order 
$q$. 

We assume that the database consists of a set of $K$ messages $X = \{\Xj{1}, \dots,\Xj{K}\}$, with each message being independently and uniformly distributed over $\GF{2^t}$ (i.e., each message $X_j$ is $t$ bits long). 
We also assume that there are $N\geq 1$ non-colluding servers 
which store identical copies of the $K$ messages.


A user is interested in downloading a message $X_W$ for some $W\in [K]$. We refer to $W$ as the \emph{demand index} and $X_W$ as the demand. 
The user has the knowledge of a subset $X_S$ of the messages 
for some $S\subset [K]$, $|S| = M$, $M<K$. 
We refer to $S$ as the \emph{side information index set} and $X_S$ 
as the \emph{side information}. 

Let $\mathbf{W}$ and $\mathbf{S}$ denote the random variables corresponding to the demand index and the side information index set. 
We restrict our attention to the class of distributions $p_{\mathbf{W}}(\cdot)$ of $\mathbf{W}$ such that $p_{\mathbf{W}}(W) > 0$ for every $W\in[K]$.

An important distribution of $\mathbf{W}$ and $\mathbf{S}$ that we focus on in this work is as follows. Let the demand index $W$ be distributed uniformly over $[K]$, i.e., 
\begin{equation}
\label{eq:WantSetDist}
p_{\mathbf{W}}(W) = \frac{1}{K},
\end{equation} for all $W\in [K]$.
Further, let the side information index set $S$ have the following conditional distribution given $W$: 
\begin{equation}\label{eq:SideInfoDist}
p_{\mathbf{S}|\mathbf{W}}(S|W) = \left\{
\begin{array}{ll}
\frac{1}{\binom{K-1}{M}}, & \textrm{if}\:\:W\not\in S \:\: \textrm{and}\:\: |S| = M,\\
0, & \textrm{otherwise}.\\
\end{array}
\right.
\end{equation}
We note that this implies  the following joint distribution on $(\mathbf{W},\mathbf{S})$: 
\begin{equation}\label{eq:dist}
p_{\mathbf{W},\mathbf{S}}(W,S) = \left\{
\begin{array}{ll}
\frac{1}{(K-M)\binom{K}{M}}, & W\not\in S,|S| = M,\\
0, & \textrm{otherwise}.\\
\end{array}
\right.
\end{equation}
We assume that the servers do not know the side information realization at the user and only know the  {\it a priori} distributions 
$p_{\mathbf{W}}(W)$
 and $p_{\mathbf{S}|\mathbf{W}}(S|W)$.


To download the message $\Xj{W}$ given the side information $\Xj{S}$, the user sends a query $\Q{W}{S}_j$ from an alphabet $\mathcal{Q}$ to the $j$-th server. 
%
The $j$-th server responds to the query it receives with an answer $\A{W}{S}_j$ over an alphabet $\mathcal{A}$. We refer to the set of queries and  answers as the {\it PIR with side information (PIR-SI) scheme}. Our focus in this paper is on non-interactive  (single round) schemes. Further, we assume that the servers do not collude with each other. A PIR-SI scheme should satisfy the following requirements.

\begin{itemize}
	\item[1.]  For every $j\in[N]$, 
    the query $\Q{W}{S}_j$ to the server $j$ 
    is a (potentially stochastic) function of $W$, $S$, and $\Xj{S}$. We assume that the answer from the server
is a deterministic function of the query and the messages, i.e.,
\begin{equation}
\label{eq:answer}
\Hc{\A{W}{S}_j}{\Q{W}{S}_j,\Xj{1}, \Xj{2}, \cdots, \Xj{K}} = 0,
\end{equation} for all $W\in [K]$, $S\subseteq[K]\setminus\{W\}$, and $j\in[N]$. 

\item[2.] From the answers $A^{[W,S]}_1,\dots,A^{[W,S]}_N$ 
and the side information $X_S$, the user should be able to decode the desired message $X_W$, i.e.,
\begin{equation}
\label{eq:decodability}
\Hc{\Xj{W}}{\A{W}{S}_1, \cdots , \A{W}{S}_N ,\Xj{S}} = 0,
\end{equation} for all $W\in [K]$, $S\subseteq[K]\setminus\{W\}$.
 
\item[3.]  The PIR-SI scheme should guarantee privacy for the user by ensuring one of the following two  conditions, referred to as  $W$-privacy and $(W,S)$-privacy as defined below. 

\begin{definition}
$W$-\textbf{privacy}: Any server cannot infer any information about the demand index from the query, answer, and messages 
i.e., for all $j\in[N]$, we have
\begin{equation}
\label{eq:privacy}
\I{\mathbf{W}}{\Q{\mathbf{W}}{\mathbf{S}}_j,\A{\mathbf{W}}{\mathbf{S}}_j,\Xj{1}, \Xj{2},\cdots,\Xj{K}} = 0.
\end{equation} 
\end{definition}

\begin{definition}
$(W,S)$-\textbf{privacy}: Any server cannot infer any information about the demand index as well as the side information index set from the query, answer, and messages
i.e., for all $j\in[N]$, we have
\begin{equation}
\label{eq:privacy2}
\I{\mathbf{W},\mathbf{S}}{\Q{\mathbf{W}}{\mathbf{S}}_j,\A{\mathbf{W}}{\mathbf{S}}_j,\Xj{1},\Xj{2},\cdots,\Xj{K}} = 0.
\end{equation}
\end{definition}

We refer to a PIR-SI scheme preserving $W$-privacy or $(W,S)$-privacy as $W$-PIR-SI or $(W,S)$-PIR-SI scheme, respectively. 
\end{itemize}

The \emph{rate} of a a PIR-SI scheme is defined as the ratio of the message length ($t$ bits) to the total length of the answers (in bits) as follows:\footnote{Note that the download rate dominates the total communication rate for large enough messages.}
\begin{equation}
\label{eq:rate}
R = \frac{t}{\sum_{j=1}^{N}\Hp{\A{W}{S}_j}}.
\end{equation}
The \emph{capacity} of $W$-PIR-SI or $(W,S)$-PIR-SI problem, respectively denoted by $C_{W}$ or $C_{W,S}$, is defined as the supremum  of rates over all $W$-PIR-SI or $(W,S)$-PIR-SI schemes for a given $N$, $K$, and $M$, respectively.

\section{Main Results}
First, we summarize our main results for single server case in Theorems \ref{thm:single-server-PIR} and~\ref{thm:single-server-PIR2}, which characterize the capacity  $W$-PIR-SI and   $(W,S)$-PIR-SI, respectively. 
\begin{theorem}
\label{thm:single-server-PIR}
For the $W$-PIR-SI problem with $N=1$ server, $K$ messages, and side information size $M$, when the demand index $\mathbf{W}$ and the side information index set $\mathbf{S}$ are jointly distributed according to~\eqref{eq:dist}, 
the capacity is
\begin{equation}
\label{eq:capacity-partition}
C_{W} = \left\lceil \frac{K}{M+1}\right\rceil^{-1}.
\end{equation}
\end{theorem}

Our proof for Theorem~\ref{thm:single-server-PIR} is based on two parts. We prove the converse in Section~\ref{sec:converse-partitioning} for any joint  distribution of $(\mathbf{W,S})$. Then, we construct an achievability scheme in Section~\ref{sec:achievability-partitioning} for the distribution given in~\eqref{eq:dist}. 


\begin{theorem}
\label{thm:single-server-PIR2}
For the $(W,S)$-PIR-SI problem with $N=1$ server storing $K$ messages 
and for any arbitrary joint distribution of the demand index $\mathbf{W}$ and the side information index set $\mathbf{S}$ such that the size of $\mathbf{S}$ is equal to $M$, the capacity is 
\begin{equation}
\label{eq:capacity-mds}
C_{W,S} = (K-M)^{-1}.
\end{equation}
\end{theorem}

First, we  show that the capacity $C_{W,S}$ of the $(W,S)$-PIR-SI problem with $N = 1$ server, $K$ messages, and size information size $M$ is upper bounded by $(K-M)^{-1}$ for any  joint distribution of the side information index set and the demand index (see 
Section~\ref{sec:converse-mds}). Further, we  construct a  scheme based on maximum distance separable (MDS) codes, which achieves this bound for any arbitrary joint distribution of $(\mathbf{W},\mathbf{S})$ such that the size of $\mathbf{S}$ is equal to $M$ (see Section~\ref{sec:achievability-mds}).  

Next, we state our main result for multiple servers storing replicas of the database, which gives a lower bound on the capacity of $W$-PIR-SI problem based on an achievability scheme.
\begin{theorem}
\label{thm:multi-server-PIR}
For the $W$-PIR-SI problem with $N$ servers, each storing $K$ messages, and side information size $M$ such that $(M+1)\mid K$, when the demand index $\mathbf{W}$ and the side information index set $\mathbf{S}$ are jointly distributed according to~\eqref{eq:dist}, 
the capacity is lower bounded as
\begin{equation}
\label{eq:capacity-partition}
C_{W} \geq \left(1 + \frac{1}{N} + \cdots + \frac{1}{N^{\frac{K}{M+1}-1}}\right)^{-1}.
\end{equation}
\end{theorem}
Our PIR scheme 
builds up on the scheme in~\cite{sun2016capacitynoncol}, which is for the case of no side-information.

\section{$W$-Privacy Problem}
\label{sec:IEEEproofs}


Our converse proofs for Theorems~\ref{thm:single-server-PIR} and~\ref{thm:single-server-PIR2} in the single-server case use the following simple yet powerful observation. 

\begin{proposition}
\label{prop:necessity}
Let $\A{W}{S}$ be an answer from the server that  satisfies the decodability requirement ~\eqref{eq:decodability} and the $W$-privacy requirement~\eqref{eq:privacy}. Then, the following two conditions hold:
\begin{enumerate}
\item For each message $X_i, i=1,\dots, K,$ there exists a subset  $\Xj{S_i}\subseteq\{\Xj{1},\cdots,\Xj{K}\} \setminus \Xj{i}$, with $|\Xj{S_i}| = M $, and a decoding function $D_{i}$ satisfying $D_i\left(\A{W}{S},\Xj{S_i}\right) = X_i$.
\item There exists a function $D_W$ such that $D_W\left(\A{W}{S},\Xj{S}\right) = X_W$. 
\end{enumerate}
\end{proposition}
\begin{proof}
The first condition is implied by the $W$-privacy requirement. Indeed, if this was not the case, then the server would know that message $X_i$ is not one of the messages requested by the user which, in turn, would violate the $W$-privacy condition (\ref{eq:privacy}). Note that the first condition holds under the assumption that $\mathbf{W}$ has a distribution such that $p_{\mathbf{W}}(W) > 0$ $\forall W\in[K]$.  

The second condition is implied by the decodability requirement.
\end{proof}

The above proposition enables us to show a relation of the PIR-SI problem with an instance of index coding with side information problem~\cite{BarYossef:IT:11,effros2015equivalence,el2010index}. We begin with briefly reviewing the index coding problem.   

\subsection{Index Coding problem}


Consider a server with $K$ messages $\Xj{1},\cdots, \Xj{K}$ of length $t$ with $\Xj{j}\in\{0,1\}^t$. Consider $L$ clients $R_1, \cdots, R_L$, $L\geq K$, where for each $i$, $R_i$ is interested in one message, denoted by $\Xj{f(i)}\in \{\Xj{i}\}$, and knows some subset $\Xj{S_i} \subset \{\Xj{i}\}$ of the other messages, referred to as its side information. 

An index code of length $\ell$ for this setting is a set of codewords in $\{0,1\}^{\ell}$ together with an encoding function $E:\{0,1\}^{tK} \rightarrow \{0,1\}^{\ell}$, and a set of $L$ decoding functions $D_1, \cdots, D_L$ such that $D_i\left(E\left(X_1,\cdots,X_K\right),X_{S_i}\right) = \Xj{f(i)}$ for all $i\in[L]$ and $[X_1, \cdots, X_K] \in \{0,1\}^{tK}$. We refer to $E\left(X_1,\cdots,X_K\right)$ as a {\it solution} to the instance of the index coding problem.

When $L = K$ and every client requires a distinct message, the side information of all the clients can be represented 
by a simple directed graph $G = \left(V,E\right)$, where $V = \{1,2,\cdots,K\}$ with the  vertex $i$ corresponding to the message $\Xj{i}$, and there is an arc $(i,j)\in E$ if $j \in S_i$. 
 We denote the out-neighbors of a vertex $i$ as $\mathcal{N}(i)$.

For a given instance of the index coding problem, the minimum encoding length $\ell$ as a function of message-length $t$ is denoted as $\beta_t$, and the {\it broadcast rate} is defined as in ~\cite{Alon:FOCS:08, Blasiak:IT:13}
\begin{equation}
\label{eq:broadcast-rate}
\beta = \inf_t \frac{\beta_t}{t} 
\end{equation}


\subsection{Converse for Theorem~\ref{thm:single-server-PIR}}
%
\label{sec:converse-partitioning}

The key step of the converse is to show that for any  scheme that satisfies the $W$-privacy constraint (\ref{eq:privacy}), the answer from the server must be a solution to an instance of the index coding problem that satisfies certain requirements as specified in the following lemma.

\begin{lemma}
\label{lem:necessary-condition}
For a $W$-PIR-SI scheme, for any demand index $W$ and side information index set $S$, the answer $\A{W}{S}$ from the server must be a solution to an instance of the index coding problem that satisfies the following requirements:
\begin{enumerate}
	\item The instance has the messages  $X_1, \cdots, X_K$;
	\item  There are $K$ clients such that each client wants to decode a distinct message from $X_1, \cdots, X_K$, and possesses a side information that includes $M$ messages;
	\item  The client that wants $X_W$ has the side information set $X_S$; for each other client the side information set has $M$ arbitrary messages from $X_1, \cdots, X_K$.
\end{enumerate}
%
%
\end{lemma}

\begin{proof}

The sets $\Xj{S_i}$ mentioned in Proposition~\ref{prop:necessity} can be used to construct the following instance of the Index Coding problem.  The instance has the message set $X_1, \cdots, X_K$ and  $K$ 
clients $\{R_1,\cdots,R_k\}$ such that:
\begin{itemize}
	\item Client $R_W$ requires packet $X_W$ and has the side information set $\Xj{S}$;
	\item Each other client $R_i,\ i\neq W$ requires $X_i$ and has side information set $\Xj{S_i}$.
\end{itemize} 
It is easy to verify that the instance satisfies all the conditions stated in the lemma and that  $\A{W}{S}$ is the feasible index code for this instance.  
\end{proof}

Note that Lemma~\ref{lem:necessary-condition} shows that the answer $\A{W}{S}$ from the server must be a solution to an instance of the index coding problem in which the out-degree of every vertex in the corresponding side information graph $G$ is equal to $M$. 
Next, we lower bound the broadcast rate for an index coding problem with side information graph $G$ such that out-degree of every vertex in $G$ is $M$ as follows.





\begin{lemma}
\label{lem:mais-lower-bound}
Let  $G$ be a directed graph on $K$ vertices such that each vertex has out-degree $M$. Then, the broadcast rate of the corresponding instance of the index coding problem is lower bounded by $\lceil \frac{K}{M+1} \rceil$.
\end{lemma}
\begin{proof}
For any side information graph $G$, the broadcast rate $\beta$ is lower bounded by the size of the maximum acyclic induced subgraph (MAIS) of $G$, denoted as $MAIS(G)$~\cite{Alon:FOCS:08,Arbabjolfae:17}.

We show that for any graph $G$ that satisfies the conditions of the lemma (i.e., the out-degree of each of the $K$ vertices of $G$ is $M$) it holds that 
$$MAIS(G)\geq \left\lceil \frac{K}{M+1} \right\rceil.$$


Specifically, we build an acyclic subgraph of $G$ induced by set $Z$ through the following procedure:


\begin{itemize}
	\item[] \hspace{-0.5cm} \textbf{Step 1.} Set $Z = \emptyset$ and  a  candidate set of vertices $V'=V$;
	\item[] \hspace{-0.5cm} \textbf{Step 2.} Add an arbitrary vertex $i\in V'$ into $Z$, i.e.,\\
$Z = Z \cup \{i\}$;
	\item[] \hspace{-0.5cm} \textbf{Step 3.} Set $V' = V' \setminus (\mathcal{N}(i) \cup \{i\})$;
	\item[] \hspace{-0.5cm} \textbf{Step 4.} There are two cases:
	\begin{itemize}
		\item[] \hspace{-0.5cm}\textbf{Case 1:} If $V' \neq \emptyset$, then repeat Steps 2-4.
		\item[] \hspace{-0.5cm}\textbf{Case 2:} If $V' = \emptyset$, then terminate the procedure and return $Z$.
	\end{itemize}
\end{itemize}



%
%

It is easy to see that the vertices in set $Z$ returned by the procedure induce an acyclic subgraph of $G$. If the vertices are ordered in the order they are added to $Z$, then there can only be an edge $(i,j)$ if $j$ was added to $Z$ before $i$. This implies that the subgraph induced by $Z$ cannot contain a cycle. 


Further, note that the set $Z$ contains at least $\lceil \frac{K}{M+1} \rceil$ vertices. At each removal step, there are at most $M+1$ vertices removed from $V$. 
Thus, the procedure iterates at least $\lceil \frac{K}{M+1} \rceil$ times, and in each iteration we add one vertex to $Z$. This implies that the size of $Z$ is at least  $\lceil \frac{K}{M+1} \rceil$.
\end{proof}


\begin{corollary}[Converse of Theorem~\ref{thm:single-server-PIR}]
For the $W$-PIR-SI problem with single server, $K$ messages, and side information size $M$, the capacity is at most ${\left\lceil \frac{K}{M+1} \right\rceil}^{-1}$.
\end{corollary}
\begin{proof}
Lemmas~\ref{lem:necessary-condition} and \ref{lem:mais-lower-bound} imply that the length of the answer $A^{[W,S]}$ is at least $t\cdot{\left\lceil \frac{K}{M+1} \right\rceil}$ bits 
for any given $W$ and $S$. Then, by \eqref{eq:rate}, it follows that $R\leq  \left\lceil \frac{K}{M+1} \right\rceil^{-1}$. 	
\end{proof}

\subsection{Achievability for Theorem~\ref{thm:single-server-PIR}}
\label{sec:achievability-partitioning}
In this section, we propose a $W$-PIR-SI scheme for $N=1$ server, $K$ messages, and side information size $M$, which achieves the rate $\left\lceil \frac{K}{M+1}\right\rceil^{-1}$.  Recall that we assume that the distribution of the demand index $W$ and the conditional distribution of the side information index set $S$ given $W$ are given respectively in~\eqref{eq:WantSetDist} and~\eqref{eq:SideInfoDist}. 
We describe the proposed scheme, referred to as the {\it Partition and Code} PIR scheme, in the following. 


{\bf Partition and Code PIR Scheme:} Given $K$, $M$, $W$, and $S$, denote $g\triangleq\left\lceil \frac{K}{M+1} \right\rceil$. The scheme consists of the following three steps.

{\it Step 1.} The user creates a partition of the $K$ messages into $g$ sets. For the ease of understanding, we describe the special case of $(M+1)\mid K$ first. 

(a) Special case of $(M+1)\mid K$: Denote $P_{1}\triangleq W \cup S$. The user randomly partitions the set of messages $[K] \setminus P_{1}$ into $g-1$ sets, each of size $M+1$, denoted as $P_2,\dots,P_{g}$.

(b) General case: 
Let $P_1,\dots,P_{g}$ be a collection of $g$ empty sets. Note that, although empty at the beginning, once constructed, the sets $P_1,\dots,P_{g-1}$ will be of size $M+1$, and the set $P_g$ will be of size $K - (g-1)(M+1)$.
The user begins by assigning probabilities to the sets according to their sizes: the sets $P_1,\dots,P_{g-1}$ are each assigned a probability $\frac{M+1}{K}$, and the set $P_g$ is assigned a probability $\frac{K - (g-1)(M+1)}{K}$. Then, the user chooses a set randomly according to the assigned probabilities of the sets. 

If the chosen set is a set $P\in \{P_1,\dots,P_{g-1}\}$, then the user fills the set $P$ with the demand index $W$ and the side information index set $S$ of the user. Next, it fills the remaining sets choosing one index at a time from the set of indices of the remaining messages uniformly at random until all the message indices are filled. 

If the chosen set is the set $P_g$, then it fill $P_g$ with the demand index $W$, and fill the remaining \mbox{$K - (g-1)(M+1)-1$} places in the set $P_g$ with randomly chosen elements from the side information index set $S$. (Note that once $P_g$ is filled, it is possible that not all of the indices in the side information index set $S$ are placed in the set.) Next, fill the remaining sets by choosing one index at a time from the set of indices of the unplaced packets uniformly at random until all packet indices are placed.

{\it Step 2.} The user sends to the server a uniform random permutation of the partition $\{P_1,\cdots,P_g\}$, ie., it sends $\{P_1, \cdots, P_g\}$ in a random order. 

{\it Step 3.} The server computes the answer $\A{W}{S}$ as a set of $g$ inner products given by $\A{W}{S} = \{A_{P_1},\dots,A_{P_{g}}\}$, where $A_P = [X_1,\dots,X_K]\cdot \mathbf{1}_{P}$ for all $P\in \{P_1,\dots,P_{g}\}$.

Upon receiving the answer from the server, the user decodes $X_W$ by subtracting off the contributions of its side information $X_S$ from $A_{P}$ for some $P\in \{P_1,\dots,P_g\}$ such that $W\in P$.

\begin{example}
Assume that $K=8$ and $M=2$. Assume that the user demands the message $X_2$ and has two messages $X_4$ and $X_6$ as side information, i.e., $W=2$ and $S=\{4,6\}$. Following the Partition and Code PIR scheme, the user labels three sets as $P_1,P_2,$ and $P_3$, and assigns probability $\frac{3}{8}$ to each of the two sets $P_1$ and $P_2$, and assigns probability $\frac{2}{8}$ to the set $P_3$. Next, the user chooses one of these sets at random according to the assigned probabilities. Assume the user has chosen the set $P_3$. The user then places $2$ into the set $P_3$, and chooses another element from $\{4,6\}$ uniformly at random to place in $P_3$ as well. Say the user chooses $6$ from the set $\{4,6\}$, then the set $P_3$ becomes $P_3 = \{2,6\}$. Then the user fills the other sets $P_1$ and $P_2$ randomly to exhaust the elements from $\{1,2,3,5,7,8\}$. Say the user chooses $P_1 = \{1,7,8\}$ and $P_2 = \{3,4,5\}$. Then the user sends to the server a random permutation of $\{\mathbf{1}_{P_1},\mathbf{1}_{P_2},\mathbf{1}_{P_3}\}$ as the query $Q^{[2,\{4,6\}]}$. The server sends three coded packets back to the user: $Y_1 = X_1 + X_7 + X_8$, $Y_2 = X_3 + X_4 + X_5$, and $Y_3 = X_2 + X_6$. The user can decode for $X_2$ by computing $X_2 = Y_3 - X_6$. From the server's perspective the user's demand is in either $\{1,7,8\}$ or $\{3,4,5\}$ with probability $\frac{3}{8}$ each, or in $\{2,6\}$ with probability $\frac{2}{8}$. The probability $P_1$ (or $P_2$) contains $W$ is $\frac{1}{3}$, and the probability that $P_3$ contains $W$ is $\frac{1}{2}$. In either case, it follows that $\mathbb{P}(\mathbf{W}=W|Q^{[1,\{2,3\}]})=\frac{1}{8}=p_{\mathbf{W}}(W)$. 
\end{example}

In the following, we show that the Partition and Code PIR scheme satisfies the $W$-privacy requirement for the setting in which the user's demand index $W$ and side information index set $S$ (given $W$) are distributed according to~\eqref{eq:WantSetDist} and~\eqref{eq:SideInfoDist}, respectively.

\begin{lemma}[Achievability of Theorem~\ref{thm:single-server-PIR}]
\label{lem:WPIRAch-NonDivis}
Consider the scenario of a $W$-PIR-SI problem in which:
\begin{itemize}
\item The server has packets $\{X_1,X_2,...,X_K\}$;
\item There is one user with $|W|=1,|S|=M$ such that $0\leq M\leq K-1$;
\item The demand index $W$ and the side information index set $S$ (given the demand index $W$) follow the distributions given in \eqref{eq:WantSetDist} and \eqref{eq:SideInfoDist}, respectively.
\end{itemize}
In this scenario, the Partition and Code PIR scheme satisfies the $W$-privacy, and has rate $R = \left\lceil \frac{K}{M+1} \right\rceil^{-1}$.
\end{lemma}

\begin{proof}
To show that the Partition and Code PIR scheme satisfies the $W$-privacy, it suffices to show that $$\mathbb{P}(\mathbf{W}=W|Q^{[W,S]})=p_{\mathbf{W}}(W).$$ 

We consider two cases as follows:
\begin{itemize}
\item[(i)] $W$ 
is in one of the sets in $\{P_1,\dots,P_{g-1}\}$. In this case, for every $i\in[g-1]$, we have
\begin{IEEEeqnarray}{rCl}
\mathbb{P}(\mathbf{W}\in P_{i}|Q^{[W,S]}) &=& \sum_{j\in P_i}\mathbb{P}(\mathbf{W} = j|Q^{[W,S]})\nonumber\\
&=& \frac{M+1}{K},\nonumber 
\end{IEEEeqnarray}
and $$\mathbb{P}(\mathbf{W}=W|\mathbf{W}\in P_i,Q^{[W,S]}) = \frac{1}{M+1}.$$ 
\item[(ii)] $W$ is the set $P_g$. In this case, $$\mathbb{P}(\mathbf{W} \in P_{g}|Q^{[W,S]})= \frac{K - (g-1)(M+1)}{K},$$ and $$\mathbb{P}(\mathbf{W}=W|\mathbf{W} \in P_{g},Q^{[W,S]}) = \frac{1}{K - (g-1)(M+1)}.$$ 
\end{itemize}
Thus, we have 
\begin{eqnarray*}
\mathbb{P}(\mathbf{W} = W|Q^{[W,S]})\hspace{65mm} \\  = \sum_{i = 1 }^{g} \mathbb{P}(\mathbf{W} = W | \mathbf{W} \in P_i,Q^{[W,S]})\mathbb{P}(\mathbf{W} \in P_i | Q^{[W,S]})\hspace{7mm}\\ 
= \frac{1}{K}. \hspace{77mm}
\end{eqnarray*}

To compute the rate of the scheme, note that 
\begin{eqnarray*} 
H(A^{[W,S]}) &=& H([A_{P_1},A_{P_2},\dots,A_{P_{g}}]) \\ &=&\sum_{P\in \{P_1,P_2,\dots,P_{g}\}} H(A_{P}) \\ &=&  t\times g,
\end{eqnarray*} 
where the equalities follow since the messages $X_j$'s (and hence the answers $A_P$'s) are independently and uniformly distributed. 
Thus, the Partition and Code PIR scheme has rate $$R = \frac{t}{t\times g} = \frac{1}{g}=\frac{M+1}{K}.$$ 
\end{proof}

\section{$(W,S)$-Privacy Problem}
\label{sec:IEEEproof-theorem-2}

In this section we consider $(W,S)$-privacy in the PIR-SI problem. We show the proof of the converse and the achievability for Theorem \ref{thm:single-server-PIR2} through a reduction to an index coding instance and an MDS coding scheme, respectively.

\subsection{Converse for Theorem~\ref{thm:single-server-PIR2}}
\label{sec:converse-mds}
When protecting the demand index and the side information index set of the user, the privacy constraint becomes $$I(\mathbf{W},\mathbf{S};Q^{[\mathbf{W},\mathbf{S}]},A^{[\mathbf{W},\mathbf{S}]},X_1,X_2,...,X_K) = 0.$$ For this case, a lower bound of $K-M$ on the number of transmissions can be shown. The proof of the converse in this case shows a necessary condition for privacy and a class of index coding problems that satisfy the necessary condition; and obtains a lower bound on the number of transmissions needed to solve the index coding problem that {belong to this class}. 

\begin{lemma}
\label{lem:WSNecessaryCondition}
For a $(W,S)$-PIR-SI scheme, for any demand index $W$ and side information index set $S$, the answer $\A{W}{S}$ from the server must be a solution to an instance of the index coding problem that satisfies the following requirements:
\begin{enumerate}
	\item The instance has the message set $X_1, \cdots, X_K$;
	\item  There are $L = (K - M)\binom{K}{M}$ clients such that each client wants to decode one message, and possesses a side information set that includes $M$ other messages;
	\item  The client that wants $X_W$ has the side information set $X_S$; for each $i\in[K], i\ne W$, for each $S_i\subset[K]\setminus\{i\}$ such that $|S_i| = M$, there exists a client that demands $X_i$ and possesses $\Xj{S_i}$ as its side information.
\end{enumerate}
\end{lemma}

\begin{proof}
Given a demand index $W$ and a side information index set $S$, let $\A{W}{S}$ be an answer from the server that satisfies the decodability requirement~\eqref{eq:decodability} and the $(W,S)$-privacy requirement~\eqref{eq:privacy2}. First, we note that the decodability requirement implies that there exists a function $D_{W,S}$ such that $D_{W,S}\left(\A{W}{S},\Xj{S}\right) = X_W$. Second, we note that the $(W,S)$-privacy requirement implies that for each message $X_i$ and every set $S_i\subseteq[K]\setminus\{i\}$ of size $M$, there exists a function $D_{i,S_i}$ satisfying $D_{i,S_i}\left(\A{W}{S},\Xj{S_i}\right) = X_i$. Otherwise, for a particular $\{i,S_i\}$, the server would know that the user cannot possess $\Xj{S_i}$ and demand $\Xj{i}$, which, in turn, would violate the $(W,S)$-privacy requirement~\eqref{eq:privacy2}.  

Now, consider an instance of the index coding problem satisfying the conditions stated in the lemma. Since  decoding functions exists for each client as argued above, $\A{W}{S}$ is a feasible index code for this instance.  
\end{proof}


Next, we give a lower bound on the broadcast rate for an instance satisfying the conditions in Lemma~\ref{lem:WSNecessaryCondition}.

\begin{lemma}
\label{lem:WSTransmissionBound}
For any instance of the index coding problem satisfying the conditions specified in Lemma~\ref{lem:WSNecessaryCondition}, the broadcast rate is at least $K - M$.
\end{lemma}

\begin{proof}
Let $J$ denote an instance of the index coding problem satisfying the conditions in Lemma~\ref{lem:WSNecessaryCondition}. Let $J'$ be an instance of the index coding problem with the $K$ messages $X_1,\cdots,X_K$ and $K-M$ clients. Each client has the side information $X_S$ and wants to decode one distinct message from $\{X_1,\cdots,X_K\}\setminus X_S$. Clearly, a solution to instance $J$ is also a solution to instance $J'$. Since the messages are independent, the broadcast rate for $J'$ is at least $K-M$, which completes the proof. 
\end{proof}



\begin{corollary}[Converse of Theorem~\ref{thm:single-server-PIR2}]
For the $(W,S)$-PIR-SI problem with $N=1$ server, $K$ messages, and side information size $M$, the capacity is at most $(K-M)^{-1}$.
\end{corollary}

\begin{proof}
Lemmas~\ref{lem:WSNecessaryCondition} and~\ref{lem:WSTransmissionBound} imply that the length of the answer $\A{W}{S}$ is at least $(K-M)t$ for any given $W$ and $S$. Thus, by using~\eqref{eq:rate}, it follows that $R\leq (K-M)^{-1}$. 
\end{proof}

\subsection{Achievability for Theorem~\ref{thm:single-server-PIR2}}
\label{sec:achievability-mds}

In this section, we give a $(W,S)$-PIR-SI scheme based on a maximum distance separable (MDS) code that achieves the rate of $1/(K - M)$. We assume that $t\geq \log_2(2K-M)$. 

\textbf{MDS PIR Scheme:} Given a demand index $W$ and a side information index set $S$ of size $M$, the user queries the server to send the $K - M$ parity symbols of a systematic $(2K - M, K)$ MDS code over the finite field $\GF{2^t}$. We assume that  $t\geq\log_2(2K-M)$, or equivalently, $2^t \geq 2K - M$. Thus, it is possible construct a $(2K - M,K)$ MDS code over $\GF{2^t}$. The answer $\A{W}{S}$ from the server consists of the $K - M$ parity symbols. 

\begin{lemma}[Achievability of Theorem~\ref{thm:single-server-PIR2}]
The MDS PIR scheme satisfies the decodability condition in~\eqref{eq:decodability} and the $(W,S)$-privacy condition in~\eqref{eq:privacy2}, and it has the rate of $R = (K-M)^{-1}$.
\end{lemma}

\begin{proof}
(Sketch) For a $(2K - M, K)$ systematic MDS code, given the $K - M$ parity symbols and any $M$ out of the $K$ messages, the user can decode all of the remaining $K - M$ messages as the code is MDS. Thus, the user can recover its demanded message. 

To ensure the $(W,S)$-privacy, note that the query and the answer are independent of the particular realization of demand index $W$ and side information index set $S$, but only depend on the size $M$ of the side information index set. As the server already knows the size of the side information index set, it does not get any other information about $W$ and $S$ from the query and the answer. Thus, the MDS PIR scheme satisfies the $(W,S)$-privacy requirement.
To compute the rate, note that for any $W$ and $S$, the answer $\A{W}{S}$ of the MDS PIR scheme consists of $K-M$ parity symbols of a $(2K-M,K)$ systematic MDS code over $\GF{2^t}$. For an MDS code, any parity symbol is a linear combination of all the messages. Thus, as each message is distributed uniformly over $\GF{2^t}$ and the code operates over $\GF{2^t}$, every parity symbol is also uniformly distributed over $\GF{2^t}$. Further, since the messages are independent, the parity symbols are independent. Hence, we have $H(\A{W}{S}) = (K - M)t$. 
Therefore, the rate of the MDS PIR scheme is $R=(K-M)^{-1}$. 
\end{proof}

\section{$W$-Privacy for Multiple Servers}
\label{sec:multi-message-pir}

\label{sec:achievability-modified-Sun-Jafar}

In this section, we present a $W$-PIR-SI scheme, when data is replicated on multiple servers. The rate achieved by the proposed scheme gives a lower bound on the capacity of multiple-server $W$-PIR-SI problem.
Our scheme builds up on the scheme proposed by Sun and Jafar in~\cite{sun2016capacitynoncol}, which deals with the case of no side information ($M=0$). We refer to it as Sun-Jafar protocol. Next, we use an example to  describe this Sun-Jafar protocol.  The details can be found in~\cite{sun2016capacitynoncol}.

\begin{example} (Sun-Jafar Protocol \cite{sun2016capacitynoncol}) $N = 2$ servers, $K = 2$ messages, and $M = 0$, i.e., no side information. The protocol assumes that each of the messages is $t = N^K = 4$ bits long. 
For a message $X_{i}$, let $[{X}_{i,1},\cdots,{X}_{i,t}]$ be a uniform random permutation of its $t$ bits. The user chooses a random permutation of the bits of $X_1$, and an independent random permutation of the bits of $X_2$. Suppose that the user is interested in downloading $X_1$. Then, it requests the bits from the first server (S1) and the second server (S2) as given in Table~\ref{tbl:example}.

\begin{table}[!h]
\begin{center}
\begin{tabular}{|c|c|}
\hline
S1 & S2\\
\hline
${X}_{1,1}$
& ${X}_{1,2}$\\ 
${X}_{2,1}$ 
& ${X}_{2,2}$\\ 
${X}_{1,3} + {X}_{2,2}$ 
& ${X}_{1,4} + {X}_{2,1}$\\ 
\hline
\end{tabular}
\end{center}
\caption{Queries for the Sun-Jafar protocol when $N = 2$ servers, $K = 2$ messages, and no side-information, when the user demands $X_1$. Each message is formed of  $4$ bits. 
}
\label{tbl:example}
\end{table}
Note that the user can decode the four bits of $X_1$ from the answers it gets. To ensure  privacy, note that each server is asked for a randomly chosen bit of each message and a sum of different pair of randomly chosen bits from each message. Therefore, a server cannot distinguish about which message is requested by the user. 
\end{example}

Next, we give an example to outline our proposed scheme for multi-server PIR with side information before describing it formally.

\begin{example} 
(Multi-Server $W$-PIR-SI Scheme) $N = 2$ servers, $K = 4$ messages, and $M = 1$ message as side information. 
Our scheme assumes that each message is $t = N^{\frac{K}{M+1}} = 4$ bits long. The  demand is privately chosen by the user, uniformly at random. The side information set has size $M = 1$. It is chosen uniformly at random from the other messages, and is unknown to the servers. 

Consider an instance when the user demands $X_1$, and the side information index  set $S = \{2\}$. First step is that the user forms a partition of $[K]$ into $g = K/(M+1) = 2$ sets $\{P_1,P_2\}$, where $P_1 = \{1,2\}$, and $P_2 = \{3,4\}$.\footnote{The general procedure for forming the partition is elaborated in the formal description of the scheme.} Next, the user sends a random permutation of $\{P_1,P_2\}$ to both the servers. The user and the servers form two {\it super-messages} by taking the sum of the messages indexed by $P_1$ and $P_2$ as follows: $\hat{X}_1 = \Xj{1} + \Xj{2}$ and $\hat{X}_2 = \Xj{3} + \Xj{4}$. 
The last step is that the user and the servers apply the Sun-Jafar protocol using the two super-messages $\hat{X}_1$ and $\hat{X}_2$, such that the user can download $\hat{X}_1$. The form of the queries is given in Table~\ref{tbl:example}.



From the answers, the user obtains $\hat{X}_1$, from which it can decode the desired message $X_1$ using the side-information $X_2$. Note that the privacy property of the Sun-Jafar protocol guarantees that any DB cannot distinguish  which super-message is requested by the user. Since the desired message can be in either super-message, and in a super-message, any of the messages can be the demand, the privacy of the demand index is ensured. 
\end{example}

Note that in the above example the proposed scheme requires to download $6$ bits, achieving the rate of $2/3$. It is shown in~\cite[Theorem 1]{sun2016capacitynoncol} that  the capacity of PIR with $N$ servers and $K$ messages and no side information is $(1 + 1/N + \cdots + 1/N^{K-1})^{-1}$. Therefore, if the user attempts to download the demand without using its side information, then the capacity is $(1 + 1/N + 1/N^2 + 1/N^3)^{-1} = 8/15$, which is smaller than $2/3$. 




Next, we describe our $W$-PIR-SI scheme for $N$ servers storing identical copies of the $K$ messages, when the user has a side information set of size $M$. We assume that $(M+1)\mid K$, and the messages are $t = N^{K/(M+1)}$ bits long. Recall that, for a subset $S \subset [K]$, $\mathbf{1}_S$ denotes the characteristic vector of the set $S$. 
Let $g\triangleq\frac{K}{M+1}$.

{\bf  Multi-Server $W$-PIR-SI Scheme:}
Assume that each message is $t = N^{\frac{K}{M+1}}$ bits long.  The scheme consists of the following three steps.

{\it Step~1.} Given the demand index $W$ and the side information index set $S$, let $P_1 = W\cup S$. The user randomly partitions the set of messages $[K] \setminus P_{1}$ into $g - 1$ sets of size $M+1$ each, denoted as $\{P_2,\cdots,P_{g}\}$. 

{\it Step~2.} The user sends to all the servers a uniform random permutation of the partition $\{P_1, \cdots, P_g\}$, ie., it sends $\{P_1, \cdots, P_g\}$ in a random order. Then, the user and the servers form $g$ {\it super-messages} $\{\hat{X}_1,\dots,\hat{X}_g\}$, where $\hat{X}_i = [X_1,\dots,X_K]\cdot\mathbf{1}_{P_i}$ for $i\in[g]$.

{\it Step~3.} The user and the $N$ servers utilize the Sun-Jafar protocol with $g$ super-messages in such a way that the user can download the message $\hat{X}_1$. 

\begin{lemma}
\label{lem:WPIRAch-NonDivis}
Consider the scenario of a $W$-PIR-SI problem in which:
\begin{itemize}
\item The $N$ servers store identical copies of $K$ messages $\{X_1,X_2,...,X_K\}$;
\item There is one user with $|W|=1,|S|=M$ such that $0\leq M\leq K-1$;
\item The demand index $W$ and the side information index set $S$ (given the demand index $W$) follow the distributions given in \eqref{eq:WantSetDist} and \eqref{eq:SideInfoDist}, respectively.
\end{itemize}
In this scenario, the multi-server $W$-PIR-SI scheme satisfies the $W$-privacy, and has rate $$R = \left(1 + 1/N + \cdots + 1/N^{K/(M+1)-1}\right)^{-1}$$
\end{lemma}
\begin{proof}
First, note that since the messages $\{X_1,\dots,X_K\}$ are uniform and independent, the super-messages $\{\hat{X}_1,\dots,\hat{X}_K\}$ are uniform and independent as well. Therefore, the rate of the scheme is that of the Sun-Jafar protocol for $N$ servers and $\frac{K}{M+1}$ messages,  which is $\left(1 + 1/N + \cdots + 1/N^{K/(M+1)-1}\right)^{-1}$, see~\cite[Theorem 1]{sun2016capacitynoncol}. 

To prove the privacy, we note that, since the Sun-Jafar protocol protects the  privacy of the demanded super-message, \ie, any server cannot have any information about which super-message the user is trying to download. Therefore, from the perspective of each server, every super-message is  equally likely to include the demanded message in the linear combination. Further, the demanded message can be any one of the $M+1$ messages in a super-message. In other words, we have
$$\mathbb{P}(\mathbf{W}\in P_i \mid \Q{W}{S}) = \frac{M+1}{K}, \quad \forall i\in[g],$$
and\\ 
$$ \mathbb{P}(\mathbf{W} = W \mid \mathbf{W}\in P_i,\Q{W}{S}) = \frac{1}{M+1},\quad\forall i\in[g].$$
Hence, we have
$$\mathbb{P}(\mathbf{W} = W|Q^{[W,S]})  = \frac{1}{K}.$$
\end{proof}

\section{Conclusion}
\label{sec:conclusion}

In this paper we considered the problem  of 
Private Information Retrieval (PIR) with side information, in which the user has {\it a priori}  a subset of the messages at the server obtained from other sources. The goal of the user is to download a message, which  is not in its side information, from the server while satisfying a certain privacy constraints. We consider two privacy requirements: $W$-privacy in which the user wants to protect the identity its demand  (i.e., which message it wants to download), and $(W,S)$-privacy in which the user wants to protect the identity of the demand and  the side information  jointly. First, we focus on the case of single server (i.e., single database). We establish the PIR capacity for $(W,S)$-privacy for arbitrary distribution of the demand index $W$ and the side information index set $S$. 
In the case of $W$-privacy, we establish the  PIR capacity  for the uniform distribution. Second, we extend our PIR scheme for $W$-privacy to the case of multiple servers (multiple copies of the database). Our scheme for the multiple servers uses ideas from the single server scheme in conjunction with the no-side-information scheme of Sun and Jafar in~\cite{sun2016capacitynoncol}. The multi-server capacities of PIR with side information under the  $W$-privacy and $(W,S)$-privacy constraints remain open.


\bibliographystyle{IEEEtran}
\bibliography{PIR_salim,coding1,coding2,pir_bib}

\end{document}